\newif\ifinformsclass
\IfFileExists{informs3.cls}{%
  \informsclasstrue
  \documentclass[moor,nonblindrev]{informs3}
}{%
  \informsclassfalse
  \documentclass[11pt]{article}
}

\ifinformsclass\else
\usepackage[margin=1in]{geometry}
\usepackage[T1]{fontenc}
\usepackage[utf8]{inputenc}
\usepackage{lmodern}
\fi
\usepackage{amsmath,amssymb,amsthm}
\usepackage{mathtools}
\usepackage{graphicx}
\usepackage{booktabs}
\usepackage{enumitem}
\usepackage{flafter}
\ifinformsclass\else
\usepackage[numbers,sort&compress]{natbib}
\fi
\usepackage[colorlinks=true,linkcolor=blue,citecolor=blue,urlcolor=blue]{hyperref}
\usepackage[nameinlink,noabbrev]{cleveref}
\graphicspath{{data/figures/}}

\crefname{section}{Section}{Sections}
\Crefname{section}{Section}{Sections}
\crefname{subsection}{Section}{Sections}
\Crefname{subsection}{Section}{Sections}
\crefname{theorem}{Theorem}{Theorems}
\Crefname{theorem}{Theorem}{Theorems}
\crefname{proposition}{Proposition}{Propositions}
\Crefname{proposition}{Proposition}{Propositions}
\crefname{corollary}{Corollary}{Corollaries}
\Crefname{corollary}{Corollary}{Corollaries}
\crefname{lemma}{Lemma}{Lemmas}
\Crefname{lemma}{Lemma}{Lemmas}
\crefname{remark}{Remark}{Remarks}
\Crefname{remark}{Remark}{Remarks}
\crefname{equation}{Equation}{Equations}
\Crefname{equation}{Equation}{Equations}
\crefname{table}{Table}{Tables}
\Crefname{table}{Table}{Tables}
\crefname{figure}{Figure}{Figures}
\Crefname{figure}{Figure}{Figures}

\renewcommand{\topfraction}{0.9}
\renewcommand{\bottomfraction}{0.7}

\renewcommand{\floatpagefraction}{0.85}

\newtheorem{theorem}{Theorem}[section]
\newtheorem{proposition}[theorem]{Proposition}

\newtheorem{corollary}[theorem]{Corollary}
\theoremstyle{remark}
\newtheorem{remark}[theorem]{Remark}

\newcommand{\proofqed}{\unskip\nobreak\hspace{0.35em}\ensuremath{\square}}

\title{Reflected UAS: Corrected Deterministic Stability and Direct CTMC Drift Calculation}
\author{Krishna Subedi\\\texttt{krishna.subedi@neryva.com}}
\date{}

\begin{document}

\maketitle

\begin{abstract}
We analyze Reflected UAS routing for heterogeneous multi-server queues at
fixed parameters under subcritical load. The deterministic surrogate is a
reflected ODE on the nonnegative orthant, not the unconstrained drift
equation. This reflected ODE has a unique boundary equilibrium characterized
by a scalar consistency equation and a convex-potential representation; all
trajectories converge to it. The older argument lifting deterministic Lyapunov
descent to CTMC stability fails: the exact generator applied to the
deterministic potential produces a boundary term absent from the reflected-ODE
descent identity. We give a direct Foster--Lyapunov drift inequality for the
CTMC using a weighted-quadratic function, bypassing the failed lift. At the
benchmark parameter point, the boundary equilibrium matches the numerical
attractor to machine precision, and the default Reflected UAS policy has lower
mean queue length than UAS and JSSQ across independent seed blocks.
\end{abstract}

\section{Introduction}

Heterogeneous multi-server queues with state-dependent routing are often
studied through deterministic surrogates: one writes down an ODE that captures
the mean drift, proves convergence, and lifts the result to the Markov chain by
a fluid-limit or Lyapunov argument
\citep{Dai1995,DaiMeyn1995,Schoenlein2015}. Reflected UAS, whose softmax
routing map is smooth and positive on the entire orthant, seems well suited to
this program. But the orthant boundary, where some queues are empty and service
completions are absent, introduces two gaps that the smoothness of the routing
map does not close.

The first gap is deterministic. Subtracting service rates from arrival rates
coordinatewise gives an unconstrained drift equation that can send trajectories
outside the nonnegative orthant at boundary states. The actual deterministic
surrogate is a reflected ODE on the nonnegative orthant, with boundary
regulation that clips outward drift to zero
\citep{HarrisonReiman1981,HeemelsSchumacherWeiland2000}. We show that under
subcritical load this reflected ODE has a unique boundary equilibrium given by a
scalar consistency equation, admits a convex-potential representation, and
attracts all trajectories.

The second gap is stochastic. A previous argument applied the deterministic
potential directly to the CTMC generator to conclude stability. That step is
invalid: the generator applied to this potential produces a nonnegative boundary
contribution that the reflected-ODE descent identity does not contain. We
identify the exact boundary term and demonstrate it on a finite-state audit.

Having diagnosed both gaps, we give a direct Foster--Lyapunov drift inequality
for the fixed-parameter CTMC. The Lyapunov function is a weighted quadratic
(not the deterministic potential), and the key bound exploits the softmax
structure to control the routing contribution through the minimum scaled queue
length. The drift constants are checked exhaustively on toy grids and against
sampled benchmark states. At the benchmark parameter point, the default
Reflected UAS policy has lower mean queue length than UAS and JSSQ across
independent seed blocks.

\Cref{sec:model} defines the model; \Cref{sec:detcore}--\Cref{sec:detval}
develop and validate the reflected-ODE theory; \Cref{sec:obstruction}
identifies the boundary obstruction;
\Cref{sec:directctmc}--\Cref{sec:benchmark} give the direct CTMC calculation
and benchmark comparison.

\section{Model and Main Questions}
\label{sec:model}

The system is a heterogeneous parallel-server queue. Jobs arrive as a Poisson
process at rate \(\lambda\) and are routed to one of \(N\) servers by a
state-dependent rule defined below. The queue-length vector is
\(Q(t)=(Q_1(t),\dots,Q_N(t))\in \mathbb{Z}_+^N\); server \(i\) works at rate
\(\mu_i>0\) and processes jobs one at a time when \(Q_i(t)>0\)
\citep{Dai1995,ChenZhang1997}.

\subsection{Reflected UAS Routing Law}
\label{subsec:routing-law}

Reflected UAS routes each arrival according to a smooth softmax rule. For a
state \(q\in \mathbb{R}_+^N\), define
\[
w_i(q)
:=
\mu_i^\gamma
\exp\!\left(-\alpha (q_i+c)/\mu_i^\beta\right),
\qquad
W(q):=\sum_{j=1}^N w_j(q),
\]
and
\begin{equation}
\label{eq:routing-probability}
p_i(q):=\frac{w_i(q)}{W(q)},
\qquad i=1,\dots,N.
\end{equation}
Here \(\alpha>0\) is an inverse-temperature parameter, \(\beta>0\) controls the
service-rate scaling inside the exponential, \(\gamma\in\mathbb{R}\) is a
service-rate prefactor exponent, and \(c\ge 0\) is a queue offset. The routing
map \(p(\cdot)\) is smooth and positive on the orthant.

For the CTMC, an arrival joins queue \(i\) with probability \(p_i(Q(t))\). The
generator acting on a test function \(f\) is
\begin{equation}
\label{eq:ctmc-generator}
(\mathcal{L}f)(Q)
=
\lambda\sum_{i=1}^N p_i(Q)\bigl[f(Q+e_i)-f(Q)\bigr]
+
\sum_{i=1}^N \mu_i \mathbf{1}_{\{Q_i>0\}}\bigl[f(Q-e_i)-f(Q)\bigr].
\end{equation}

\subsection{Deterministic Surrogate}
\label{subsec:deterministic-surrogate}

The unconstrained drift associated with the routing law is
\[
F(q):=\lambda p(q)-\mu,
\qquad
\mu:=(\mu_1,\dots,\mu_N).
\]
Setting \(\dot q = F(q)\) ignores nonnegativity: this drift can point outside
the orthant when some queue is empty. The deterministic model is the reflected
ODE on \(\mathbb{R}_+^N\), defined in \Cref{sec:detcore}
\citep{HarrisonReiman1981,HeemelsSchumacherWeiland2000,NagurneyZhang1996}.

Throughout the paper we keep two questions separate: the behavior of the
reflected ODE as a deterministic system, and whether this ODE arises as the
fluid limit of the CTMC.

\subsection{Anchor Benchmark}

One parameter point serves as the empirical and computational anchor
throughout:
\[
(\alpha,\beta,\gamma,c)=(20,0.85,0.5,0.5),
\]
\[
\mu=(0.5,0.7,0.9,1.1,1.3,1.5,1.7,1.9,2.1,2.3),
\qquad
\lambda=11.2.
\]
Since \(\Lambda:=\sum_{i=1}^N \mu_i = 14.0 > \lambda\), the system is
subcritical. This point appears in the deterministic validation
(\Cref{sec:detval}), the CTMC drift checks (\Cref{sec:directctmc}), and the
policy comparison against UAS and JSSQ
\citep{Houck1987,SelenAdanKapodistriaVanLeeuwaarden2016}.

\section{Deterministic Reflected-ODE Theory}
\label{sec:detcore}

\subsection{Reflected Dynamics}

Let \(C:=\mathbb{R}_+^N\). Recall the unconstrained drift
\(F(q):=\lambda p(q)-\mu\) from Section~\ref{subsec:deterministic-surrogate}.
The reflected dynamics is
\[
q_i(t)
=
q_i(0)+\int_0^t F_i(q(s))\,ds + y_i(t),
\qquad i=1,\dots,N,
\]
where each regulator \(y_i\) is nondecreasing, \(y_i(0)=0\), and satisfies the
complementarity condition
\[
\int_0^\infty q_i(s)\,dy_i(s)=0.
\]
In coordinates,
\begin{equation}
\label{eq:reflected-coordinate-dynamics}
\dot q_i=
\Gamma_i(q):=
\begin{cases}
F_i(q), & q_i>0,\\
\max\{F_i(q),0\}, & q_i=0.
\end{cases}.
\end{equation}
On the interior this is \(\dot q = \lambda p(q)-\mu\); at the boundary,
reflection clips outward drift to zero
\citep{HarrisonReiman1981,HeemelsSchumacherWeiland2000,NagurneyZhang1996}.

\begin{proposition}
\label{prop:no-interior-equilibrium}
Assume \(\lambda<\Lambda:=\sum_{i=1}^N \mu_i\). Then the reflected ODE has no
equilibrium \(q^*\in(0,\infty)^N\) with all coordinates strictly positive.
\end{proposition}

\begin{proof}
If \(q^*\in(0,\infty)^N\), then every regulator is locally constant at
equilibrium, so
\[
0=\lambda p_i(q^*)-\mu_i,
\qquad i=1,\dots,N.
\]
Summing over \(i\) gives
\[
0=\lambda\sum_{i=1}^N p_i(q^*)-\sum_{i=1}^N \mu_i
=\lambda-\Lambda,
\]
which contradicts \(\lambda<\Lambda\).\proofqed
\end{proof}

\begin{corollary}
\label{cor:boundary-equilibrium}
Under \(\lambda<\Lambda\), every equilibrium of the reflected ODE lies on a
boundary face of \(C\).
\end{corollary}

\subsection{Exact Boundary Equilibrium}

For each server \(i\), define the threshold constant
\[
\theta_i:=\mu_i^{\gamma-1}\exp\!\left(-\alpha c/\mu_i^\beta\right).
\]
At equilibrium, reflection imposes
\[
q_i^*\ge 0,
\qquad
\lambda p_i(q^*)\le \mu_i,
\qquad
q_i^*(\mu_i-\lambda p_i(q^*))=0.
\]
Let
\[
w_i^*:=\mu_i^\gamma \exp\!\left(-\alpha (q_i^*+c)/\mu_i^\beta\right),
\qquad
W^*:=\sum_{j=1}^N w_j^*,
\qquad
K:=\frac{W^*}{\lambda}.
\]
Then \(p_i(q^*)=w_i^*/W^*\), so
\[
q_i^*>0 \Rightarrow w_i^*=\mu_i K,
\qquad
q_i^*=0 \Rightarrow w_i^*\le \mu_i K.
\]
Hence
\[
q_i^*=
\max\!\left\{
0,\frac{\mu_i^\beta}{\alpha}\log\!\left(\frac{\theta_i}{K}\right)
\right\}.
\]
Substituting back into \(W^*=\lambda K\) gives the scalar consistency equation
\begin{equation}
\label{eq:scalar-consistency}
\lambda = G(K):=\sum_{i=1}^N \min\!\left\{\mu_i,\frac{\mu_i\theta_i}{K}\right\}.
\end{equation}

\begin{theorem}
\label{thm:exact-boundary-equilibrium}
Assume \(\lambda<\Lambda\). Then the equation \(\lambda=G(K)\) has a unique
solution \(K^*>0\). The reflected ODE has a unique equilibrium \(q^*\), and it
is given coordinatewise by
\[
q_i^*=
\max\!\left\{
0,\frac{\mu_i^\beta}{\alpha}\log\!\left(\frac{\theta_i}{K^*}\right)
\right\},
\qquad i=1,\dots,N.
\]
Moreover,
\[
q_i^*>0 \Longleftrightarrow K^*<\theta_i.
\]
\end{theorem}

\begin{proof}
Each term in \(G\) is continuous and nonincreasing on \((0,\infty)\), so \(G\)
is continuous and nonincreasing. Also,
\[
\lim_{K\downarrow 0} G(K)=\Lambda>\lambda,
\qquad
\lim_{K\uparrow\infty} G(K)=0<\lambda.
\]
Hence at least one solution exists. Let
\(\theta_{\min}:=\min_{1\le i\le N}\theta_i\). For \(0<K\le \theta_{\min}\),
every term in \(G\) equals \(\mu_i\), so \(G(K)=\Lambda\). Therefore any
solution must satisfy \(K>\theta_{\min}\). On \((\theta_{\min},\infty)\), at
least one term is strictly decreasing and every other term is nonincreasing, so
\(G\) is strictly decreasing there. The solution is therefore unique. The
coordinate formula and active-set characterization follow from the
complementarity relations.\proofqed
\end{proof}

\subsection{Convex-Gradient Structure}

Define the diagonal matrix
\[
D:=\operatorname{diag}(\mu_1^\beta,\dots,\mu_N^\beta),
\]
and the potential
\begin{equation}
\label{eq:deterministic-potential}
H(q):=
\sum_{i=1}^N \mu_i^{1-\beta} q_i
+
\frac{\lambda}{\alpha}\log\!\left(
\sum_{j=1}^N \mu_j^\gamma
\exp\!\left(-\alpha (q_j+c)/\mu_j^\beta\right)
\right).
\end{equation}

\begin{proposition}
\label{prop:gradient-identity}
For every \(q\in C\),
\[
\partial_i H(q)=\frac{\mu_i-\lambda p_i(q)}{\mu_i^\beta},
\qquad i=1,\dots,N.
\]
Equivalently,
\[
\lambda p(q)-\mu = -D\nabla H(q).
\]
\end{proposition}

\begin{proof}
Differentiate the log-normalizing term:
\[
\partial_i \log W(q)
=
-\frac{\alpha}{\mu_i^\beta}p_i(q).
\]
Differentiate the linear term of \(H\) and combine the two contributions.\proofqed
\end{proof}

\begin{proposition}
\label{prop:convexity-coercivity}
If \(\lambda<\Lambda\), then \(H\) is convex and coercive on \(C\).
\end{proposition}

\begin{proof}
The Hessian satisfies
\[
v^\top \nabla^2 H(q)v
=
\lambda\alpha
\left[
\sum_{i=1}^N p_i(q)\left(\frac{v_i}{\mu_i^\beta}\right)^2
-\left(\sum_{i=1}^N p_i(q)\frac{v_i}{\mu_i^\beta}\right)^2
\right],
\]
which is nonnegative because the bracket is a variance under the probability
vector \(p(q)\). Thus \(H\) is convex.

For coercivity, write \(s_i=q_i/\mu_i^\beta\) and
\(a_i=\mu_i^\gamma e^{-\alpha c/\mu_i^\beta}\). Then
\[
H(q)=\sum_{i=1}^N \mu_i s_i
+
\frac{\lambda}{\alpha}\log\!\left(\sum_{i=1}^N a_i e^{-\alpha s_i}\right).
\]
Using the probability vector \(\rho_i=\mu_i/\Lambda\) and the lower bound
\(\log\sum_i e^{x_i}\ge \sum_i \rho_i x_i\), one obtains
\[
H(q)\ge \left(1-\frac{\lambda}{\Lambda}\right)\sum_{i=1}^N \mu_i s_i + C
\]
for a finite constant \(C\). Since \(\lambda<\Lambda\), the coefficient of the
linear term is positive, so \(H(q)\to+\infty\) as \(\|q\|_2\to\infty\) with
\(q\in C\).\proofqed
\end{proof}

\begin{corollary}
\label{cor:minimizer-equals-equilibrium}
The constrained problem
\[
\min_{q\in C} H(q)
\]
has a unique minimizer, and that minimizer is the equilibrium \(q^*\) from
Theorem~\ref{thm:exact-boundary-equilibrium}.
\end{corollary}

\begin{proof}
Convexity and coercivity give existence of a minimizer. The first-order
condition \(0\in \nabla H(q)+N_C(q)\) is exactly the equilibrium
complementarity condition because of Proposition~\ref{prop:gradient-identity}.
Uniqueness then follows from Theorem~\ref{thm:exact-boundary-equilibrium}.\proofqed
\end{proof}

\subsection{Global Convergence}

Along a reflected trajectory \(q(\cdot)\),
\[
\frac{d}{dt}H(q(t))
=
\sum_{i=1}^N \partial_i H(q(t))\,\dot q_i(t)
=
-\sum_{i:\,\dot q_i(t)\neq 0}\mu_i^\beta \bigl(\partial_i H(q(t))\bigr)^2
\le 0
\]
for almost every \(t\). Equality holds if and only if the equilibrium
complementarity conditions are satisfied. Since \(H\) is coercive, trajectories
are bounded; convergence then follows by standard semigroup arguments
\citep{Dunn1981,Brezis1973,Bruck1975,Pazy1978}.

\begin{theorem}
\label{thm:global-convergence}
Fix \(\alpha>0\), \(\beta>0\), \(\gamma\in\mathbb{R}\), \(c\ge 0\), service
rates \(\mu_i>0\), and \(\lambda<\Lambda\). Then the reflected ODE on
\(\mathbb{R}_+^N\) has the unique equilibrium \(q^*\) from
Theorem~\ref{thm:exact-boundary-equilibrium}, and every reflected trajectory satisfies
\[
\lim_{t\to\infty} q(t)=q^*.
\]
\end{theorem}

\begin{proof}
The Lyapunov descent identity gives that \(H(q(t))\) is nonincreasing. By
Proposition~\ref{prop:convexity-coercivity}, the trajectory is bounded and
precompact.
Every accumulation point must satisfy the equilibrium complementarity
conditions; otherwise the descent identity would stay strictly negative in a
neighborhood of that point. By Theorem~\ref{thm:exact-boundary-equilibrium},
the only such point is \(q^*\). Hence the full trajectory converges to
\(q^*\).\proofqed
\end{proof}

\begin{remark}
This convergence is deterministic and does not imply positive Harris
recurrence of the CTMC; \Cref{sec:obstruction} identifies why.
\end{remark}

\section{Benchmark Deterministic Validation}
\label{sec:detval}

We check the closed-form equilibrium from
Theorem~\ref{thm:exact-boundary-equilibrium} against the numerical attractor
of the reflected ODE. For each system, we solve \(\lambda=G(K)\),
reconstruct \(q^*\), and compare with the integrator output. The diagnostics
are the sup-norm discrepancy and active-set match.

\begin{table}[htbp]
  \centering
  \caption{Boundary equilibrium verification diagnostics.}
  \label{tab:boundary_equilibrium}
  \begin{tabular}{lrrrr}
    \toprule
    System & $N$ & $\rho$ & $\|q^*_{\text{exact}} - q^*_{\text{ODE}}\|_\infty$ & Active Set \\
    \midrule
    Benchmark & 10 & 0.8000 & 2.22e-14 & 5/10 \\
    Symmetric 4-server & 4 & 0.8000 & 0.00e+00 & 0/4 \\
    Heavy-load 5-server & 5 & 0.9500 & 2.89e-15 & 4/5 \\
    Light-load 3-server & 3 & 0.2000 & 0.00e+00 & 0/3 \\
    UAS special case & 10 & 0.8000 & 9.99e-15 & 6/10 \\
    \bottomrule
  \end{tabular}
\end{table}

At the ten-server benchmark (\Cref{sec:model}), \(K^*=2.9632\times 10^{-4}\),
the active set has \(5\) servers, and the maximum coordinatewise discrepancy
is \(2.22\times 10^{-14}\).

Four additional systems (symmetric, high-load asymmetric, light-load, and
UAS-parameter) all pass with active-set agreement and maximum discrepancy
at or below \(9.99\times 10^{-15}\).

\section{Why the Old Stochastic Lift Fails}
\label{sec:obstruction}

The convergence theorem from \Cref{sec:detcore} does not imply positive Harris
recurrence of the CTMC. The gap appears in two places: the fluid scaling and
the exact generator applied to \(H\).

\subsection{The Scaling Gap}

A classical Dai-style fluid argument would require the same deterministic object
to arise from the usual queue-length scaling
\[
\bar Q^{(r)}(t):=\frac{1}{r}Q^{(r)}(rt).
\]
See \citet{Dai1995,DaiMeyn1995} for the fluid-limit stability argument. For
Reflected UAS,
\[
p_i(q)=
\frac{\mu_i^\gamma \exp\!\left(-\alpha (q_i+c)/\mu_i^\beta\right)}
{\sum_{j=1}^N \mu_j^\gamma \exp\!\left(-\alpha (q_j+c)/\mu_j^\beta\right)}.
\]
If one inserts a fluid-scale state \(rq\), then
\[
p_i(rq)=
\frac{\mu_i^\gamma \exp\!\left(-\alpha (rq_i+c)/\mu_i^\beta\right)}
{\sum_{j=1}^N \mu_j^\gamma \exp\!\left(-\alpha (rq_j+c)/\mu_j^\beta\right)}.
\]
With fixed \(\alpha>0\), the factor \(rq_i\) appears inside the exponent. As
\(r\to\infty\), the routing map hardens toward an argmin-type rule on the
scaled coordinates instead of remaining the smooth softmax map used in
\Cref{sec:detcore}. The reflected ODE is not yet identified with the classical
fluid limit of the CTMC.

\subsection{Exact Generator Calculation}

Let \(Q\in \mathbb Z_+^N\). The CTMC generator is
\[
(\mathcal L f)(Q)
=
\lambda\sum_{i=1}^N p_i(Q)\bigl[f(Q+e_i)-f(Q)\bigr]
+
\sum_{i=1}^N \mu_i \mathbf 1_{\{Q_i>0\}}\bigl[f(Q-e_i)-f(Q)\bigr].
\]
Apply \(\mathcal L\) to the deterministic potential
\[
H(Q)=
\sum_{i=1}^N \mu_i^{1-\beta}Q_i
+
\frac{\lambda}{\alpha}\log W(Q),
\qquad
W(Q)=\sum_{j=1}^N \mu_j^\gamma
\exp\!\left(-\alpha (Q_j+c)/\mu_j^\beta\right).
\]
Define \(a_i:=\alpha/\mu_i^\beta\). Then the exact one-step increments are
\[
H(Q+e_i)-H(Q)
=
\mu_i^{1-\beta}
+
\frac{\lambda}{\alpha}
\log\!\Bigl[1-p_i(Q)\bigl(1-e^{-a_i}\bigr)\Bigr],
\]
and, for \(Q_i>0\),
\[
H(Q-e_i)-H(Q)
=
-\mu_i^{1-\beta}
+
\frac{\lambda}{\alpha}
\log\!\Bigl[1+p_i(Q)\bigl(e^{a_i}-1\bigr)\Bigr].
\]

\begin{proposition}
\label{prop:generator-boundary-mismatch}
The exact generator drift of \(H\) can be written as
\[
(\mathcal L H)(Q)
=
-\sum_{i:\,Q_i>0}\mu_i^\beta \bigl(\partial_i H(Q)\bigr)^2
+
\sum_{i:\,Q_i=0}\lambda p_i(Q)\partial_i H(Q)
+
R(Q),
\]
where \(R(Q)\) is uniformly bounded on \(\mathbb Z_+^N\).
\end{proposition}

\begin{proof}
From Proposition~\ref{prop:gradient-identity},
\[
\partial_i H(Q)=\frac{\mu_i-\lambda p_i(Q)}{\mu_i^\beta}.
\]
By the Hessian formula from Proposition~\ref{prop:convexity-coercivity},
\(\partial_{ii}^2 H(q)\le \lambda\alpha/\mu_i^{2\beta}\) for every
\(q\in\mathbb R_+^N\), so the diagonal second derivatives of \(H\) are
uniformly bounded on the orthant. Apply first-order Taylor expansion along each
coordinate:
\[
H(Q+e_i)-H(Q)=\partial_i H(Q)+r_i^+(Q),
\]
\[
H(Q-e_i)-H(Q)=-\partial_i H(Q)+r_i^-(Q)
\qquad (Q_i>0).
\]
Hence \(r_i^\pm(Q)\) are uniformly bounded. Substituting these expansions into
the generator gives
\[
(\mathcal L H)(Q)
=
\sum_{i=1}^N \bigl(\lambda p_i(Q)-\mu_i\mathbf 1_{\{Q_i>0\}}\bigr)\partial_i H(Q)
+
R(Q),
\]
with uniformly bounded \(R(Q)\). For \(Q_i>0\),
\[
\lambda p_i(Q)-\mu_i=-\mu_i^\beta \partial_i H(Q),
\]
which yields the negative square term. For \(Q_i=0\), the service jump is
absent, so only the arrival contribution
\(\lambda p_i(Q)\partial_i H(Q)\) remains.\proofqed
\end{proof}

\subsection{The Boundary Obstruction}

Proposition~\ref{prop:generator-boundary-mismatch} differs from the
reflected-ODE Lyapunov identity at the boundary. In the reflected ODE, a
boundary coordinate with \(\partial_i H(q)\ge 0\) is clipped and contributes
\(0\) to \(\dot H\).
In the CTMC, the same boundary coordinate contributes
\(\lambda p_i(Q)\partial_i H(Q)\), which is strictly positive whenever
\(Q_i=0\), \(p_i(Q)>0\), and \(\partial_i H(Q)>0\).

This boundary contribution is why the old lift fails. The obstruction is
specific: it blocks direct application of \(H\) to the generator, not all
Lyapunov approaches.

\begin{table}[htbp]
  \centering
  \caption{CTMC generator boundary mismatch diagnostics.}
  \label{tab:boundary_mismatch}
  \resizebox{\textwidth}{!}{%
  \begin{tabular}{lrrrrl}
    \toprule
    System & $N$ & Boundary States & Positive $\partial$-term & Max Gap & Obstruction \\
    \midrule
    Mismatch 2-server symmetric & 2 & 31 & 1 & 1.5113 & Yes \\
    Mismatch 2-server asymmetric & 2 & 31 & 0 & 3.6809 & No \\
    Mismatch 3-server & 3 & 361 & 0 & 8.8257 & No \\
    \bottomrule
  \end{tabular}
  }
\end{table}

In the symmetric two-server system, one of \(31\) boundary states has a
strictly positive boundary contribution (maximum \(3.2\times 10^{-1}\)). The
other toy systems show no positive boundary term in this audit, so the
obstruction appears as a specific witness rather than a universal pattern.

\begin{figure}[!htbp]
\centering
\includegraphics[width=0.86\textwidth]{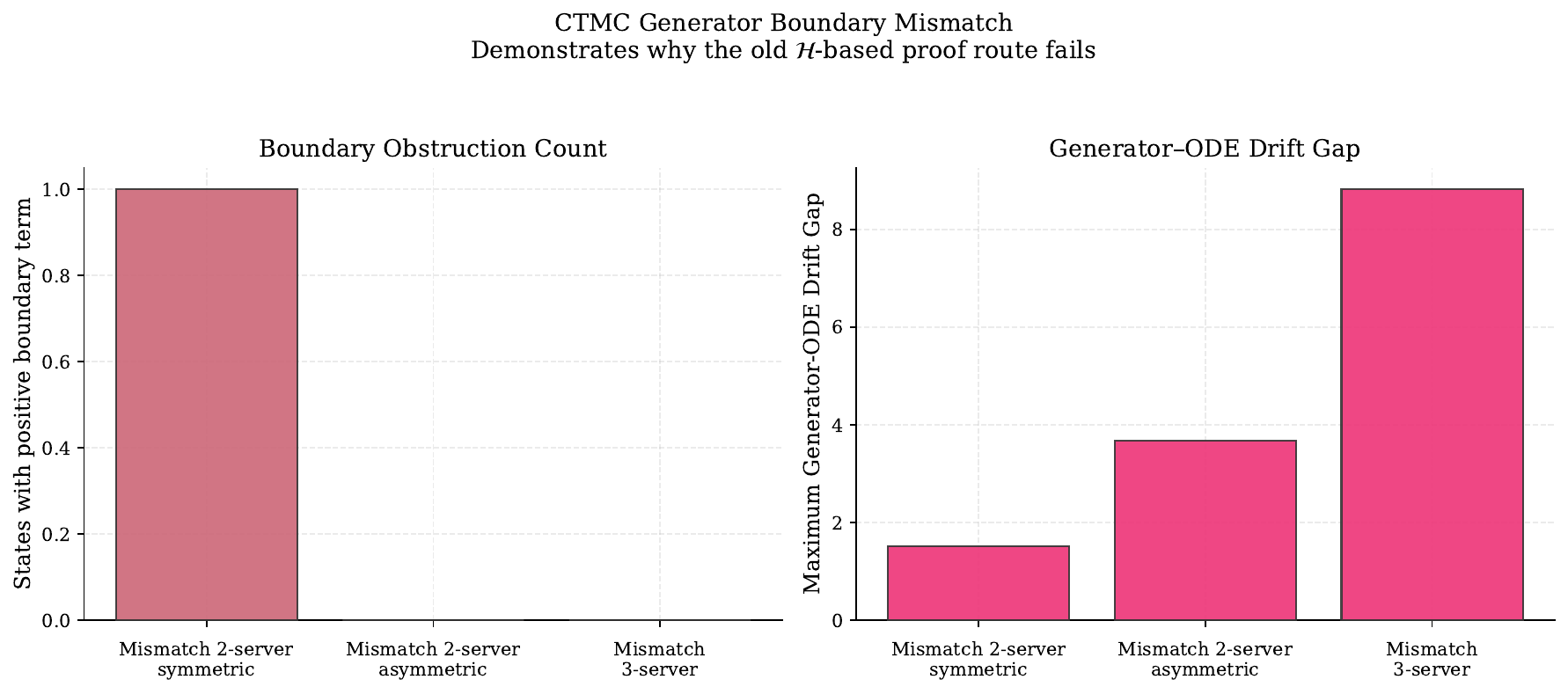}
\caption{Toy-system obstruction summary for the old \(H\)-based lift. The left
panel records whether the audit exhibits a positive boundary contribution. The
right panel records the size of the generator--ODE drift gap.}
\label{fig:h3_mismatch_demo}
\end{figure}

\begin{remark}
Any stochastic stability proof must either bound the boundary contribution in
Proposition~\ref{prop:generator-boundary-mismatch} or use a Lyapunov function
other than \(H\).
\end{remark}

For the CTMC, a Lyapunov function must handle both one-step generator jumps
and boundary states where service completions are absent. The weighted
quadratic in \Cref{sec:directctmc} does this: its generator expansion needs
only a bound on the softmax contribution to the minimum scaled queue length.

\section{Direct CTMC Drift Calculation}
\label{sec:directctmc}

We construct a Foster--Lyapunov drift inequality for the CTMC using a
weighted-quadratic Lyapunov function and a softmax minimum bound
\citep{MeynTweedie1993}.

\subsection{Weighted Quadratic Drift}

Consider the norm-like function
\[
V(Q):=\frac12\sum_{i=1}^N \frac{Q_i^2}{\mu_i^\beta},
\qquad Q\in\mathbb Z_+^N.
\]
For an arrival to coordinate \(i\),
\[
V(Q+e_i)-V(Q)=\frac{Q_i+1/2}{\mu_i^\beta},
\]
and for a service completion at coordinate \(i\) with \(Q_i>0\),
\[
V(Q-e_i)-V(Q)=-\frac{Q_i-1/2}{\mu_i^\beta}.
\]
Therefore the generator satisfies
\[
(\mathcal L V)(Q)
=
\lambda\sum_{i=1}^N p_i(Q)\frac{Q_i}{\mu_i^\beta}
-
\sum_{i=1}^N \mu_i^{1-\beta}Q_i
+
R_0(Q),
\]
where
\[
R_0(Q)
:=
\frac{\lambda}{2}\sum_{i=1}^N p_i(Q)\mu_i^{-\beta}
+
\frac12\sum_{i=1}^N \mu_i^{1-\beta}\mathbf 1_{\{Q_i>0\}}
\]
is uniformly bounded.

Rewrite the routing law as an ordinary softmax over shifted
energies:
\[
s_i(Q):=\frac{Q_i}{\mu_i^\beta}+\kappa_i,
\qquad
\kappa_i:=\frac{c}{\mu_i^\beta}-\frac{\gamma}{\alpha}\log\mu_i.
\]
Then
\[
p_i(Q)=\frac{e^{-\alpha s_i(Q)}}{\sum_{j=1}^N e^{-\alpha s_j(Q)}}.
\]
The entropy variational identity gives
\[
\sum_{i=1}^N p_i(Q)\frac{Q_i}{\mu_i^\beta}\le m(Q)+C_1,
\qquad
m(Q):=\min_{1\le i\le N}\frac{Q_i}{\mu_i^\beta},
\]
for a finite constant \(C_1\) depending only on
\((\alpha,\beta,\gamma,c,\mu)\) \citep{WainwrightJordan2008}. Substituting this
bound into the generator and decomposing
\[
Q_i=\mu_i^\beta m(Q)+\Delta_i(Q),
\qquad
\Delta_i(Q)\ge 0,
\]
yields
\[
(\mathcal L V)(Q)
\le
-(\Lambda-\lambda)m(Q)
-
\sum_{i=1}^N \mu_i^{1-\beta}\Delta_i(Q)
+
R,
\]
with \(R<\infty\). Since
\[
|Q|_1=\left(\sum_{i=1}^N \mu_i^\beta\right)m(Q)+\sum_{i=1}^N \Delta_i(Q),
\]
one may choose
\[
\varepsilon
:=
\min\!\left\{
\frac{\Lambda-\lambda}{\sum_{i=1}^N \mu_i^\beta},
\min_{1\le i\le N}\mu_i^{1-\beta}
\right\}
>0
\]
and obtain the linear drift inequality
\[
(\mathcal L V)(Q)\le -\varepsilon |Q|_1 + R.
\]

\begin{proposition}
\label{prop:conditional-direct-ctmc}
Fix \(\alpha>0\), \(\beta>0\), \(\gamma\in\mathbb R\), \(c\ge 0\), service
rates \(\mu_i>0\), and \(\lambda<\Lambda:=\sum_{i=1}^N \mu_i\). Under the
weighted-quadratic calculation above, the Reflected-UAS queue-length CTMC
admits a Foster-Lyapunov drift inequality of the form
\[
(\mathcal L V)(Q)\le -\varepsilon |Q|_1 + R
\]
for some \(\varepsilon>0\) and \(R<\infty\).
\end{proposition}

\begin{remark}
Proposition~\ref{prop:conditional-direct-ctmc} bypasses the obstruction from
\Cref{sec:obstruction} by using a different Lyapunov function and drift
decomposition.
\end{remark}

\subsection{Validation Checks}

Each step is validated on the toy grids and benchmark state bank: the
generator identity, shifted-energy representation, softmax bound, and drift
inequality.

At the benchmark parameter point,
\[
\varepsilon=0.2128,
\qquad
R=24.34,
\]
with no drift violations in the sampled state bank (worst residual
\(-17.887\)). Five auxiliary parameter points all have positive
\(\varepsilon\) and no violations; a sweep of \(294\) points, including the
benchmark default, satisfies the same checks.

\begin{table}[htbp]
  \centering
  \caption{Exhaustive small-grid CTMC drift diagnostics.}
  \label{tab:drift_audit}
  \begin{tabular}{lrrrrrr}
    \toprule
    System & $N$ & $\rho$ & States & Violations & $\varepsilon$ & $R$ \\
    \midrule
    Toy 2-server symmetric & 2 & 0.8000 & 496 & 0 & 0.2000 & 1.86 \\
    Toy 2-server asymmetric & 2 & 0.8000 & 496 & 0 & 0.2141 & 2.91 \\
    Toy 3-server & 3 & 0.8000 & 5456 & 0 & 0.2136 & 4.45 \\
    Toy 2-server UAS & 2 & 0.8000 & 496 & 0 & 0.2000 & 1.91 \\
    \bottomrule
  \end{tabular}
\end{table}

The exhaustive check (\Cref{tab:drift_audit}) evaluates every state in the toy
grids. All four systems pass with zero violations and positive
\(\varepsilon\). These checks verify the implementation, not the theorem.

\begin{figure}[!htbp]
\centering
\includegraphics[width=0.82\textwidth]{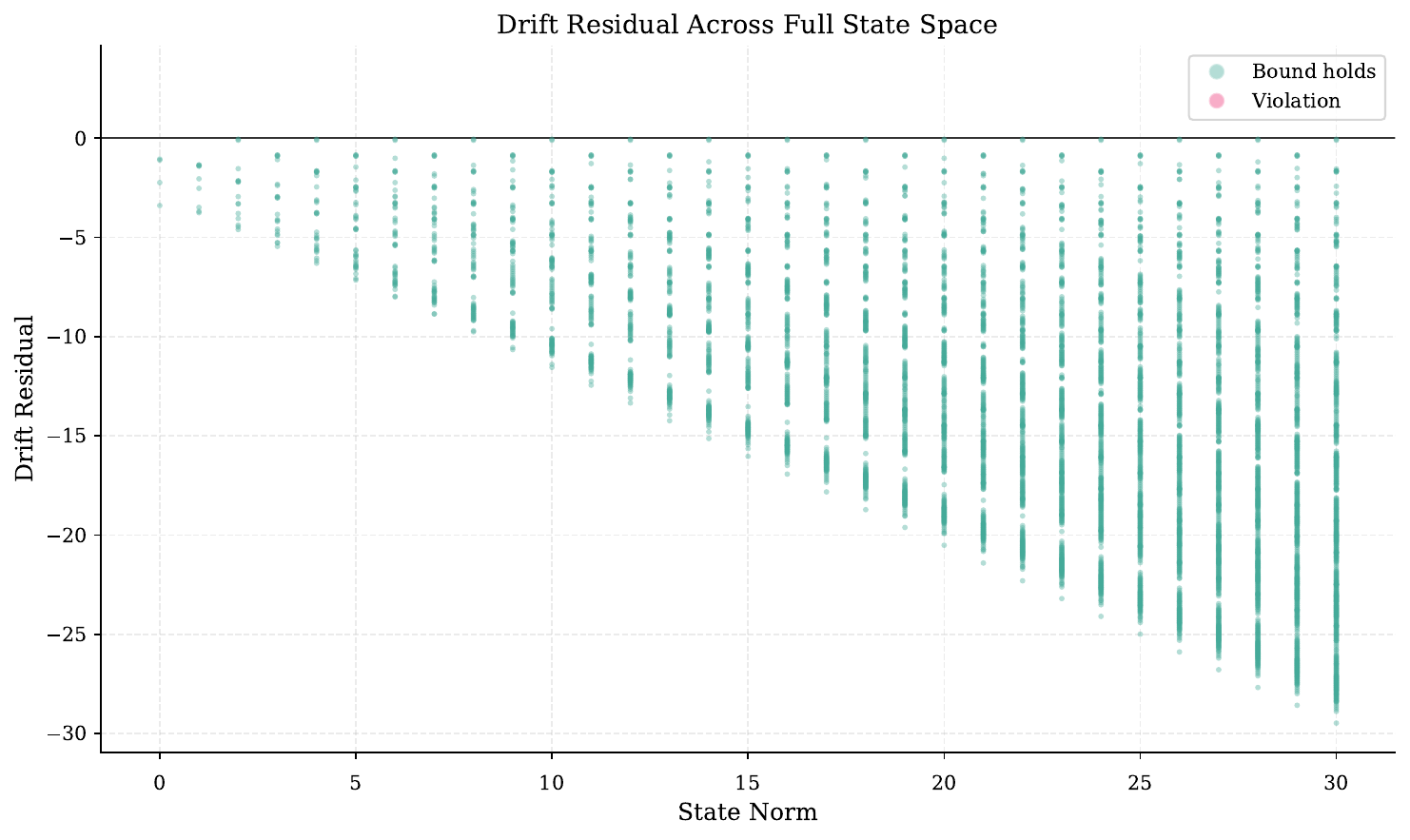}
\caption{Statewise residuals for the exhaustive toy-grid drift check. Every
point lies below zero, so the drift inequality holds on the selected
finite-state toy grids.}
\label{fig:h4_drift_grid}
\end{figure}

\section{CTMC Validation and Benchmark Comparison}
\label{sec:benchmark}

The benchmark comparison evaluates Reflected UAS against two baselines: UAS
(the unmodified rule) and JSSQ, a shortest-expected-delay dispatch rule for
heterogeneous servers
\citep{Winston1977,Houck1987,GuptaHarcholBalterSigmanWhitt2007,SelenAdanKapodistriaVanLeeuwaarden2016}.
The drift calculation from \Cref{sec:directctmc} applies at this parameter
point.

\begin{table}[htbp]
  \centering
  \caption{Benchmark empirical performance averaged across seed blocks.}
  \label{tab:benchmark_empirical}
  \begin{tabular}{lrrr}
    \toprule
    Policy & $\mathbb{E}[|Q|_1]$ & Gini Index & Mean Sojourn \\
    \midrule
    Reflected UAS (default) & 10.04 & 0.1740 & 0.896 \\
    JSSQ & 11.00 & 0.3222 & 0.982 \\
    UAS & 11.47 & 0.3286 & 1.024 \\
    \bottomrule
  \end{tabular}
\end{table}

\Cref{tab:benchmark_empirical} summarizes the independent-seed rerun over seed
blocks \(1000\), \(2000\), and \(3000\). Reflected UAS has the lowest mean
steady-state total queue length, Gini index, and mean sojourn time. Averaged
across blocks, mean total queue length is \(10.04\) for Reflected UAS,
\(11.00\) for JSSQ, and \(11.47\) for UAS.

The pairwise differences are consistent across all three seed blocks: Reflected
UAS is lower than JSSQ by about \(0.96\) units and lower than UAS by about
\(1.43\) units.

\begin{figure}[!htbp]
\centering
\includegraphics[width=0.82\textwidth]{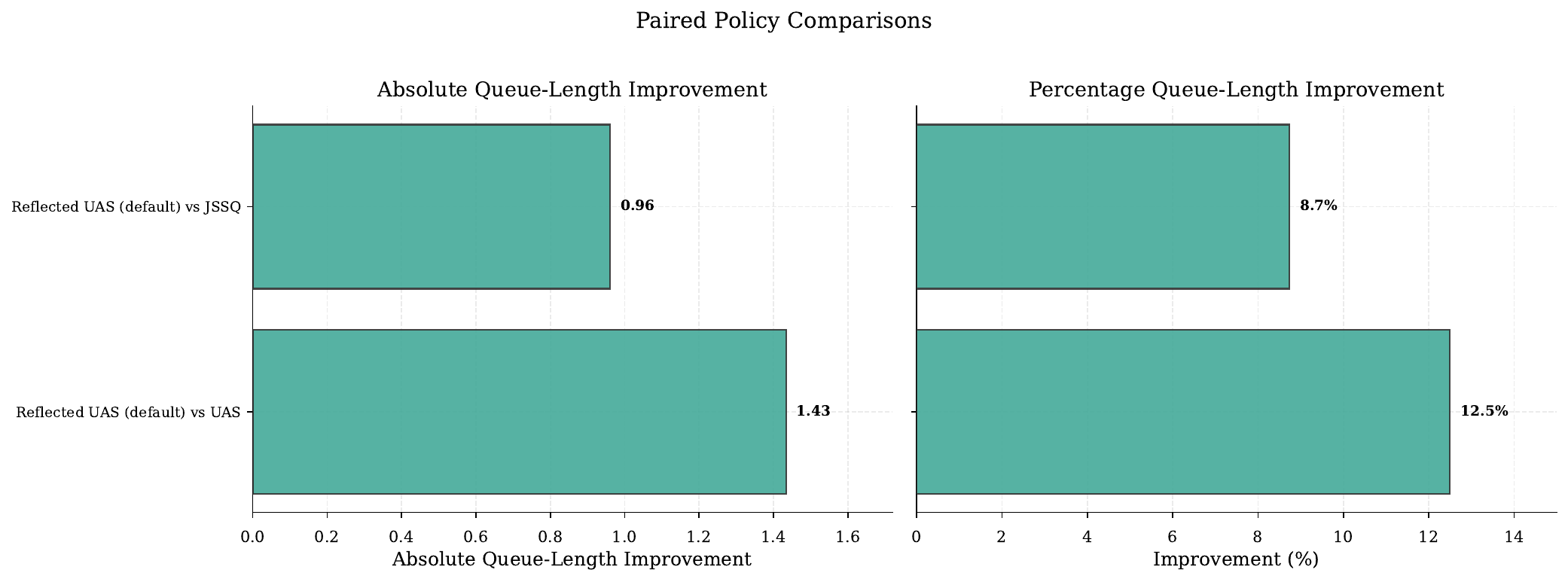}
\caption{Reported benchmark comparison of Reflected UAS against JSSQ and UAS.
The reflected benchmark policy has lower mean total queue length than both
baselines in the reported seed blocks.}
\label{fig:h5_benchmark_comparisons}
\end{figure}

Because Reflected UAS is softmax-based, it defines a differentiable routing
map compatible with gradient-based policy optimization, unlike hard dispatch
rules such as JSQ (join-the-shortest-queue) and JSSQ
\citep{Williams1992,SuttonMcAllesterSinghMansour1999}.

\section{Discussion and Scope}
\label{sec:discussion}

The reflected ODE under subcritical load has a unique boundary equilibrium and
a convex Lyapunov function that drives all trajectories to it. This
deterministic theory is complete. The stochastic picture is not: the
deterministic potential \(H\) does not transfer to the CTMC generator because
of the boundary term in \Cref{sec:obstruction}, and the weighted-quadratic
drift inequality from \Cref{sec:directctmc} is verified only on sampled
states, not the full lattice.

Closing the stochastic gap requires either a proof that the softmax minimum
bound holds uniformly, or a different Lyapunov construction that avoids the
boundary obstruction. A fluid-limit identification would also suffice, but the
fixed-\(\alpha\) softmax hardens under fluid scaling
(\Cref{sec:obstruction}), so a direct Dai-style argument does not apply.

The benchmark comparison is empirical and limited to one parameter point with
three seed blocks. It does not establish dominance of Reflected UAS over UAS
or JSSQ in general.

\appendix
\section{Supplementary Derivations and Numerical Evidence}
\label{sec:appendix}
\makeatletter
\setlength{\@fptop}{0pt plus 1pt}
\setlength{\@fpbot}{0pt plus 3fil}
\setcounter{topnumber}{3}
\setcounter{bottomnumber}{2}
\setcounter{totalnumber}{4}
\renewcommand{\topfraction}{0.95}
\renewcommand{\bottomfraction}{0.85}
\renewcommand{\floatpagefraction}{0.80}
\makeatother

The appendix gives the derivations and numerical checks behind the main text.
The mathematical statements used by the paper are contained in the manuscript;
the code and data release reproduce the tables and figures.

\newcommand{\apptablecaption}[2]{%
  \refstepcounter{table}%
  \noindent\textbf{Table \thetable.} #1\label{#2}\par\medskip
}

\newcommand{\appfigurecaption}[2]{%
  \refstepcounter{figure}%
  \noindent\textbf{Figure \thefigure.} #1\label{#2}\par\smallskip
}

\subsection{Deterministic Equilibrium and Convergence Checks}

The deterministic calculations compare the closed-form equilibrium in
\Cref{thm:exact-boundary-equilibrium} with independently computed
reflected-ODE trajectories. The diagnostics report the terminal diameter of
multi-start trajectories, the terminal vector-field residual, the
supremum-norm discrepancy between the closed-form and numerical equilibria, and
the active set size.

\begin{table}[htbp]
  \centering
  \caption{Reflected-ODE multi-start convergence diagnostics.}
  \label{tab:reflected_ode_convergence}
  \begin{tabular}{lrrrr}
    \toprule
    System & $N$ & Trajectories & Diameter & $\|\dot{q}(T)\|_\infty$ \\
    \midrule
    Benchmark & 10 & 34 & 1.22e-15 & 2.66e-15 \\
    Symmetric 4-server & 4 & 28 & 0.00e+00 & 0.00e+00 \\
    \bottomrule
  \end{tabular}
\end{table}

\begin{center}
\includegraphics[width=0.82\textwidth]{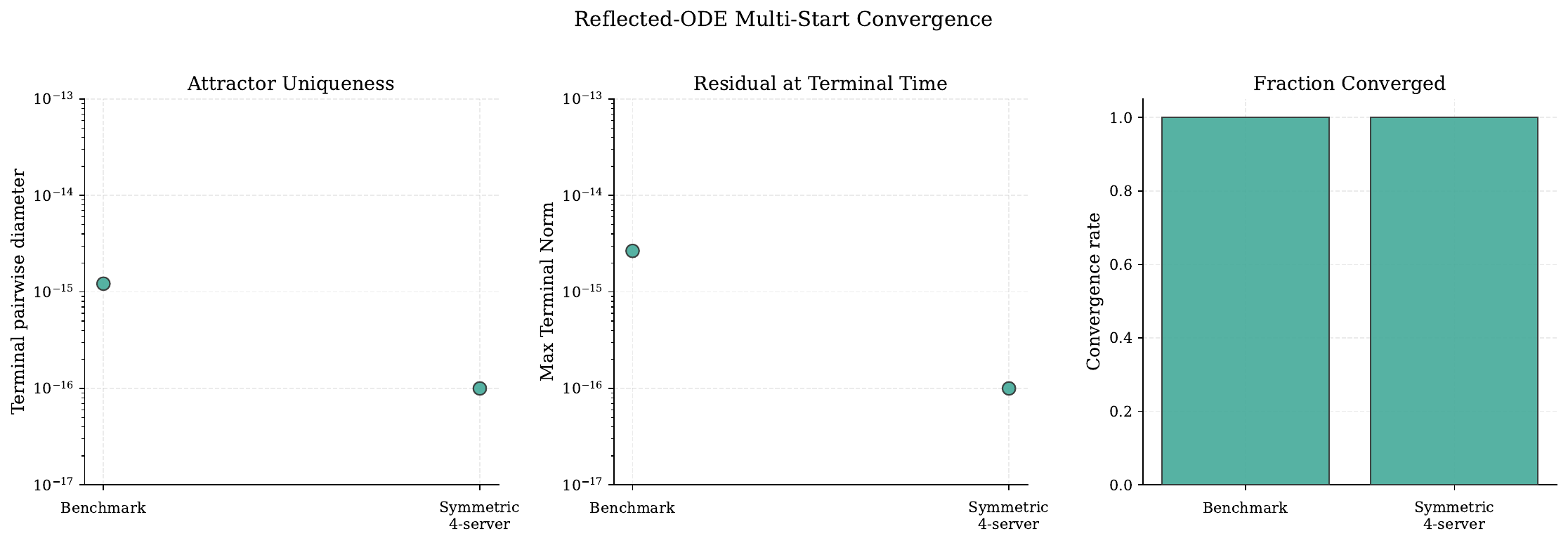}
\end{center}
\appfigurecaption{Multi-start convergence diagnostics for the reflected ODE. The panels
report terminal trajectory diameter, terminal residual norm, and the fraction
of trajectories that reach the common attractor.}{fig:appendix_h1h2_summary}

\begin{center}
\includegraphics[width=0.82\textwidth]{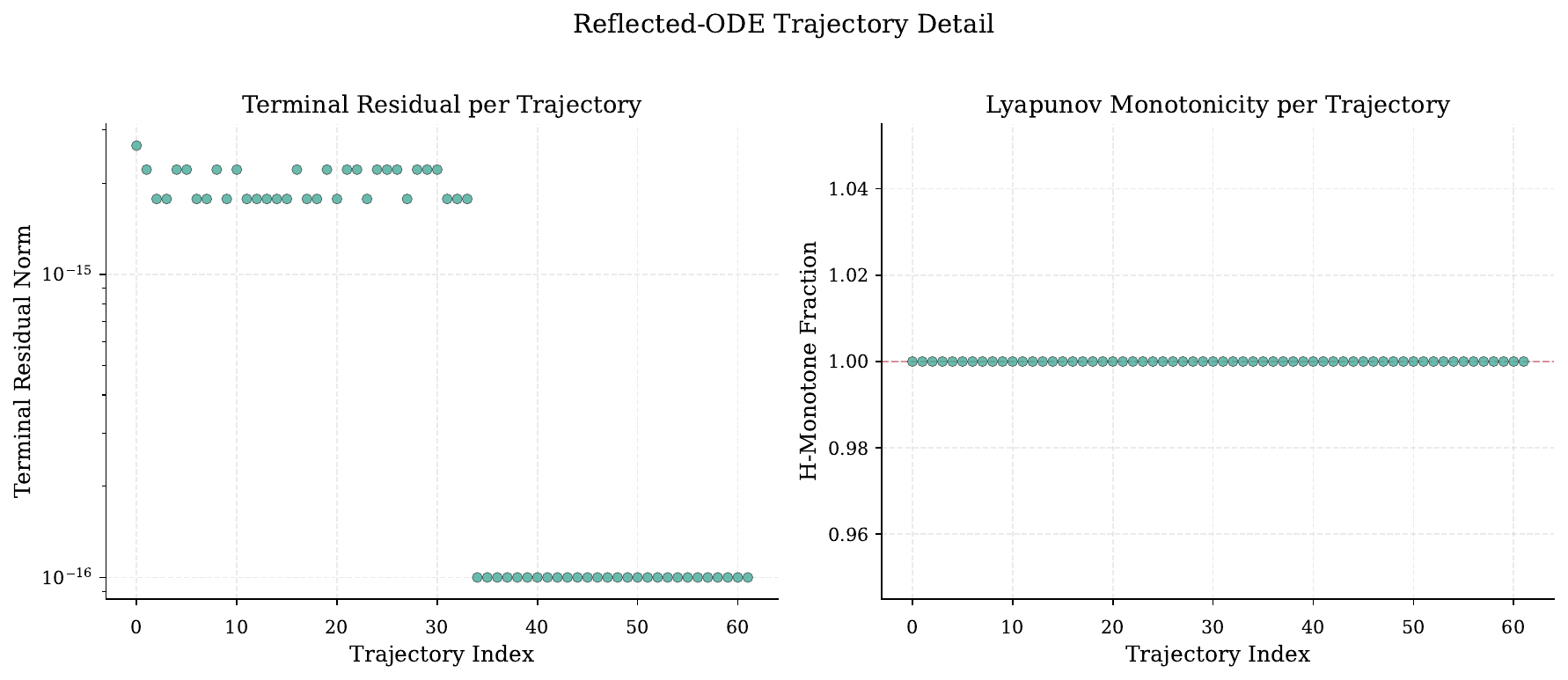}%
\vspace{-4pt}
\end{center}
\appfigurecaption{Trajectory-level reflected-ODE diagnostics. Each point corresponds to
one trajectory and reports the terminal residual norm together with the
observed Lyapunov monotonicity along the numerical path.}{fig:appendix_h1h2_trajectories}

\begin{center}
\includegraphics[width=0.72\textwidth]{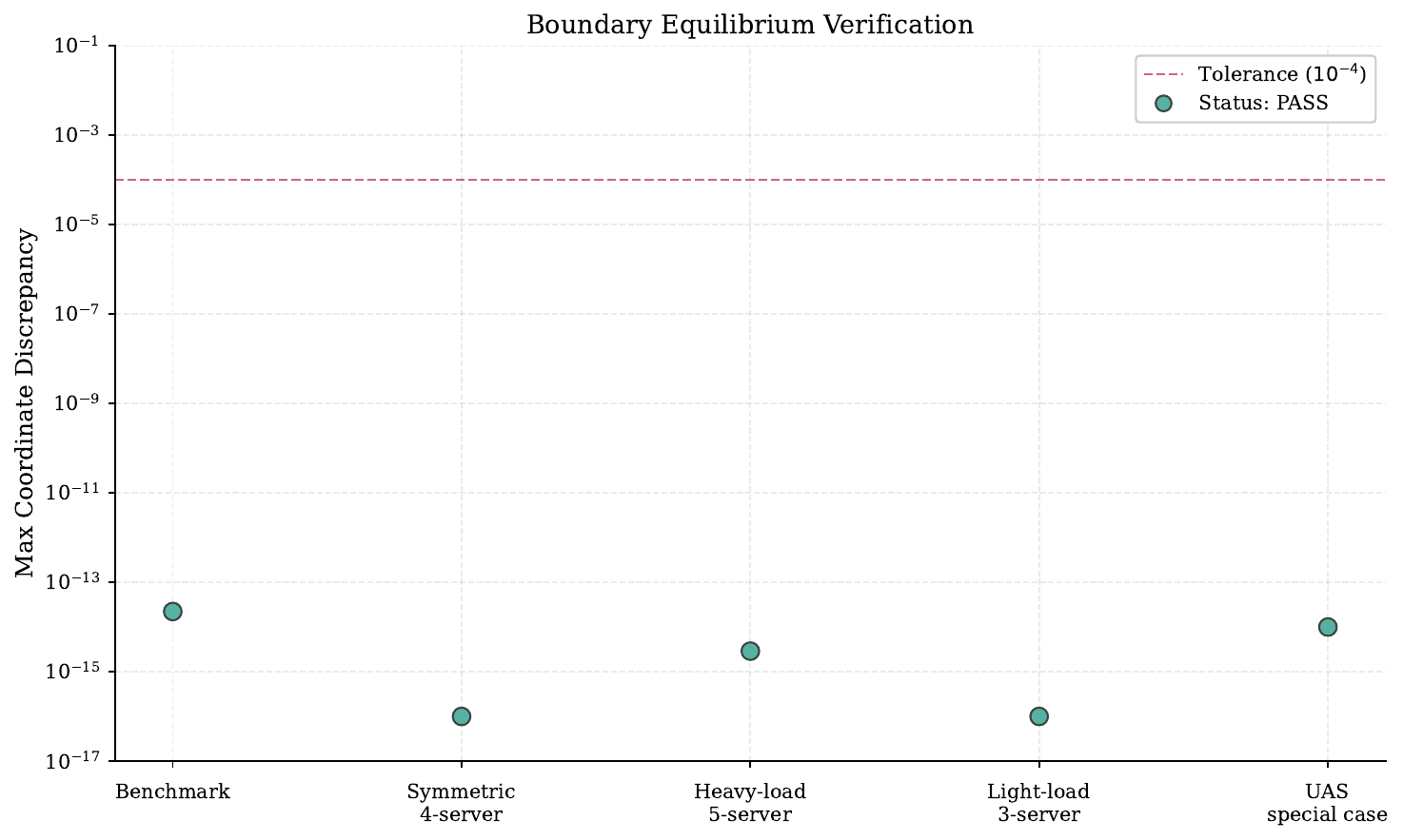}%
\vspace{-4pt}
\end{center}
\appfigurecaption{Supremum-norm discrepancy between the closed-form boundary
equilibrium and the numerical reflected-ODE attractor. The tolerance line is
shown for scale.}{fig:appendix_h2_boundary_verification}

\subsection{Boundary Obstruction Diagnostics}

The obstruction diagnostics evaluate the exact CTMC generator drift of the
deterministic potential \(H\) at boundary states. The state-level decomposition
separates the interior negative-square contribution from the boundary term in
\Cref{prop:generator-boundary-mismatch}, identifying the term that is absent
from the reflected-ODE descent identity.

\begin{center}
\includegraphics[width=0.58\textwidth]{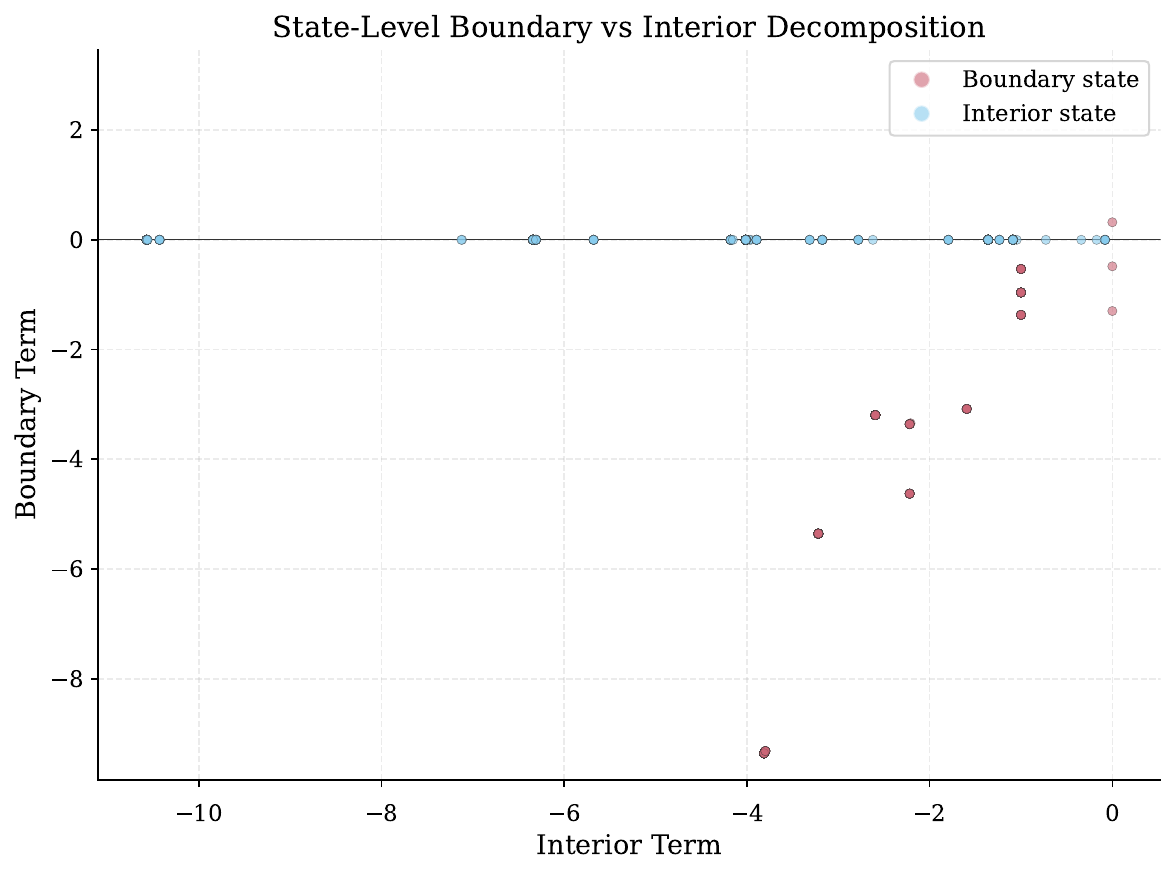}%
\vspace{-4pt}
\end{center}
\appfigurecaption{State-level decomposition of the generator mismatch in the two-server
obstruction example. The axes show the interior contribution and the boundary
contribution to the drift of \(H\).}{fig:appendix_h3_states}

\subsection{Finite-State CTMC Drift Checks}

The finite-state CTMC calculations evaluate the weighted-quadratic drift
inequality from \Cref{sec:directctmc} on selected toy grids and benchmark state
sets. They check the generator identity, the softmax minimum bound, and the
reported constants on the displayed finite collections of states.

\begin{center}
\includegraphics[width=0.78\textwidth]{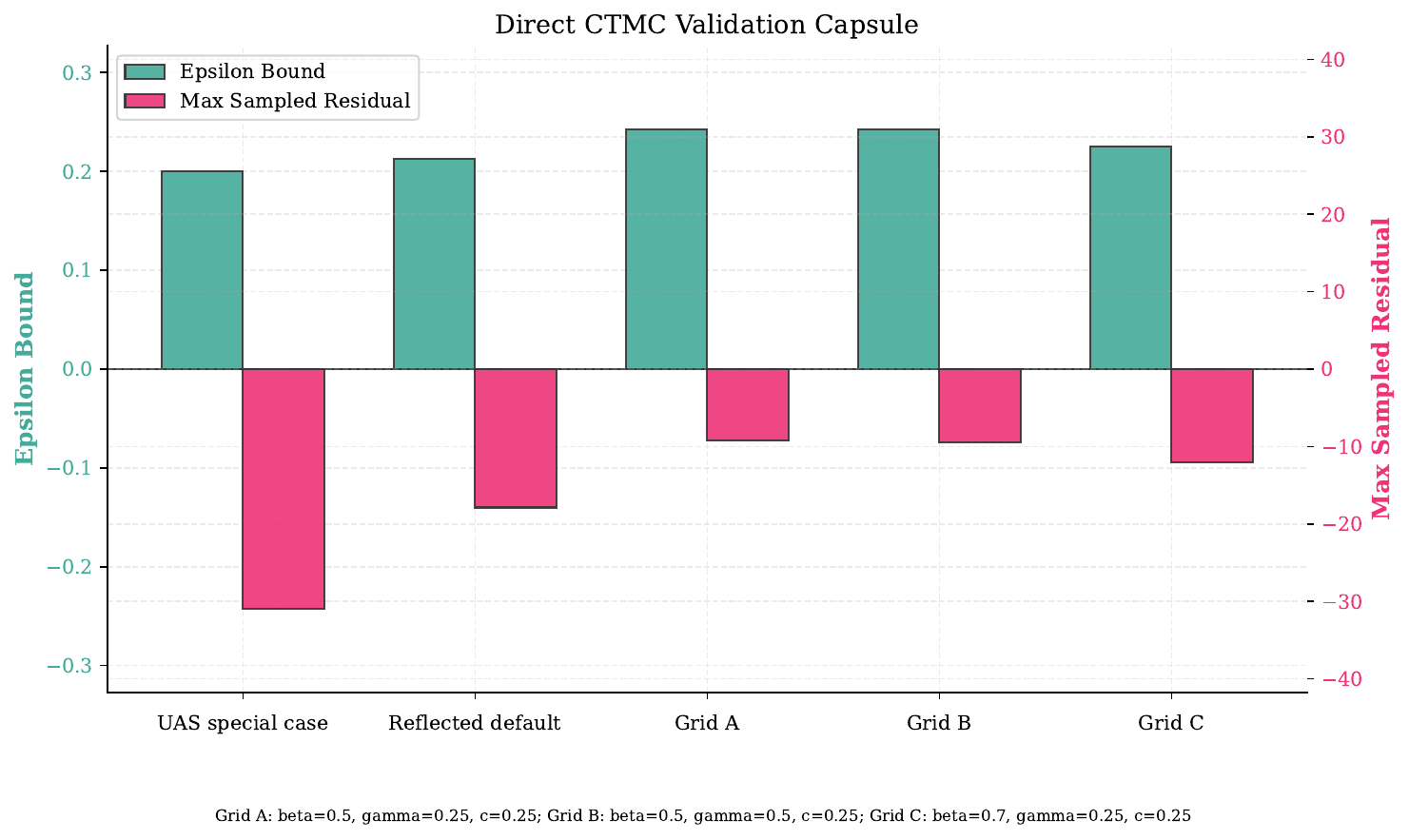}%
\vspace{-4pt}
\end{center}
\appfigurecaption{Direct CTMC drift diagnostics for the candidate parameter family. The
figure reports the computed drift constants and residuals for the displayed
candidate points.}{fig:appendix_h4_direct_audit}

\begin{center}
\includegraphics[width=0.65\textwidth]{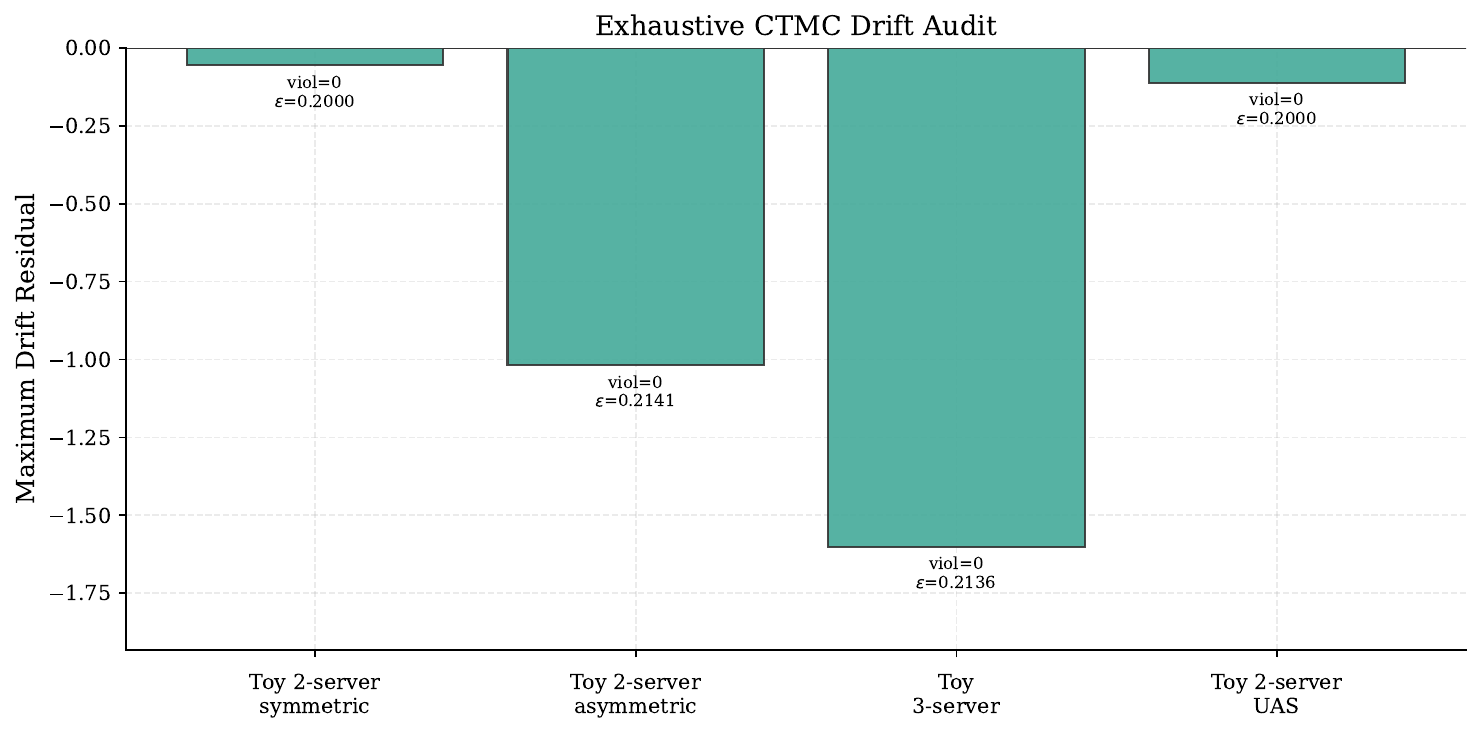}%
\vspace{-4pt}
\end{center}
\appfigurecaption{Exhaustive toy-grid drift check. For each toy system, the figure
reports the number of examined states, the number of drift-inequality
violations, and the computed constants \(\varepsilon\) and \(R\).}{fig:appendix_h4_exhaustive_summary}

\vspace{-6pt}
\subsection{Broader Applicability: Differentiable Policy Learning}

One advantage of Reflected UAS is its smooth, softmax-based routing map.
Unlike hard dispatch rules such as JSQ or JSSQ, the Reflected UAS formulation
is differentiable. The routing policy can be used in gradient-based learning
workflows, including neural-network-driven reinforcement learning
\citep{Williams1992,SuttonMcAllesterSinghMansour1999}.

The main text analyzes the fixed-parameter queueing model. The diagnostics
below show how the smooth routing map can serve as a differentiable policy
class for learning-oriented control of heterogeneous queues.

\begin{center}
\includegraphics[width=0.95\textwidth,height=0.42\textheight,keepaspectratio]{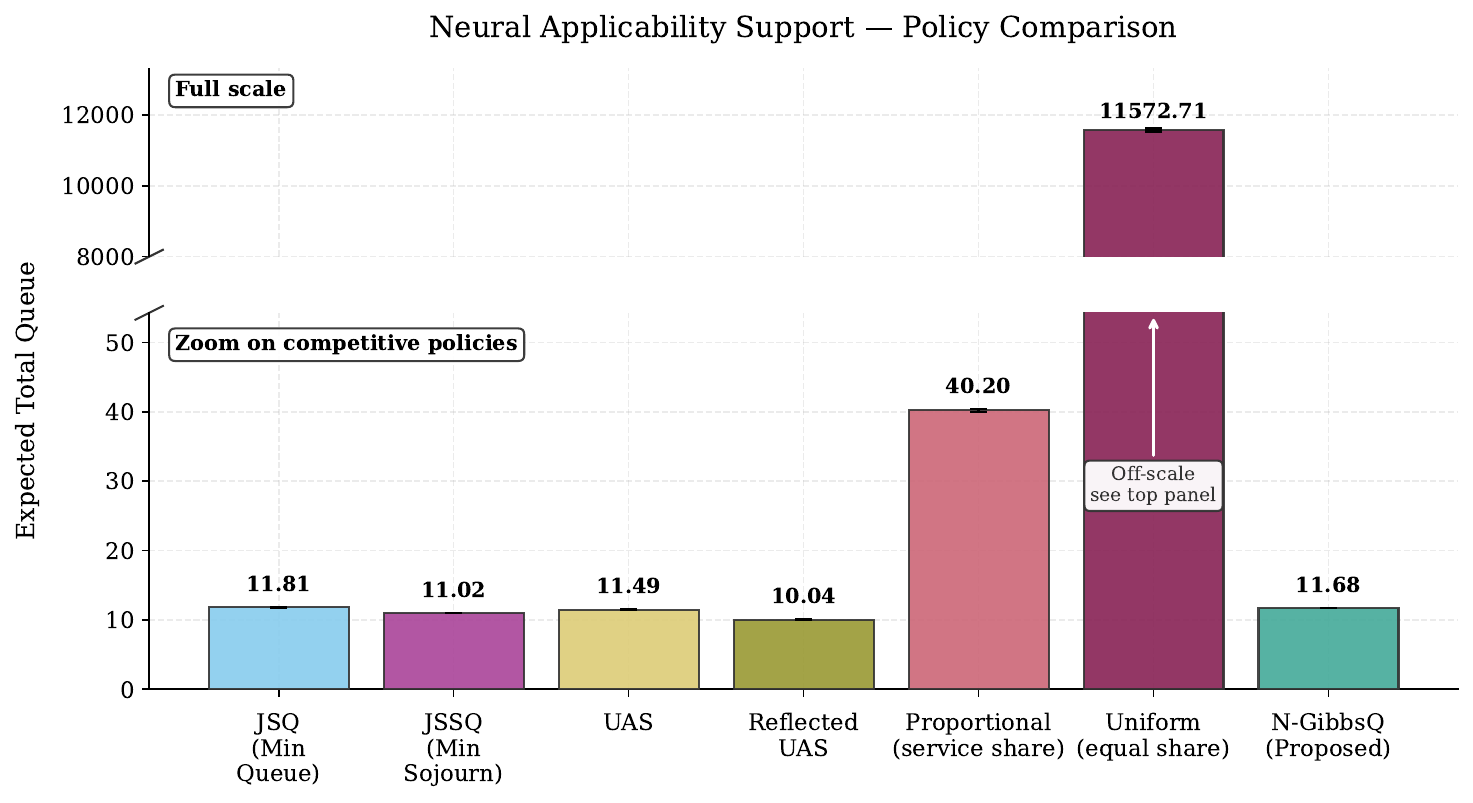}%
\vspace{-4pt}
\end{center}
\appfigurecaption{Reflected UAS in a differentiable-policy learning
environment. The smooth routing map allows gradient-based optimization while
preserving the queue-aware structure of the policy.}{fig:appendix_h7_support}

\begin{center}
\includegraphics[width=0.82\textwidth]{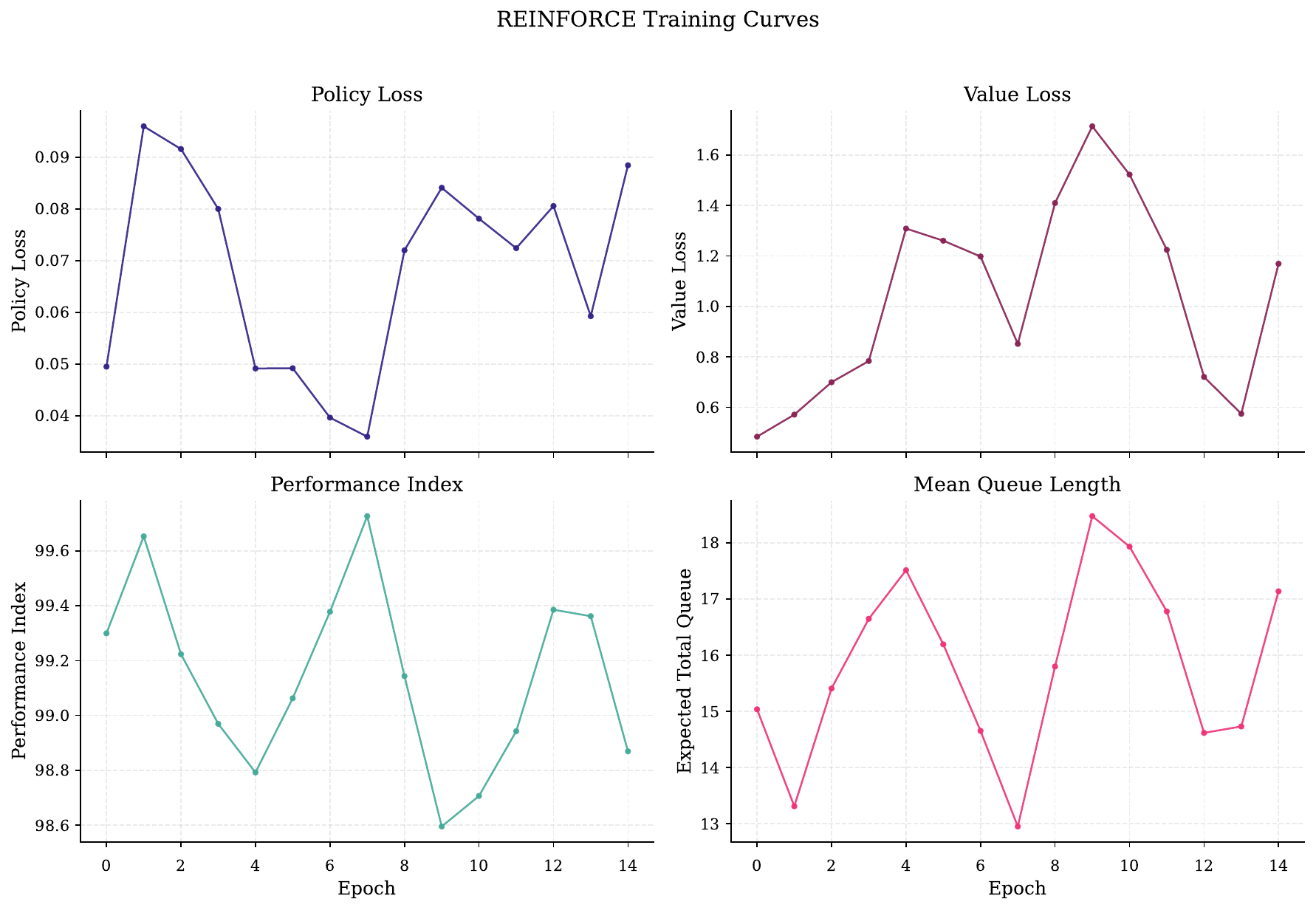}%
\vspace{-4pt}
\end{center}
\appfigurecaption{Training diagnostics for gradient-based policy optimization. The
continuous Reflected UAS policy parametrization yields stable learning
trajectories in the displayed reinforcement-learning run.}{fig:appendix_h7_training}

\subsection{Reproducibility Note}

The numerical results in the paper are reproducible from the project
repository:
\url{https://github.com/neryva/gibbsq}.
The repository contains the parameter files, simulation drivers, manuscript
sources, and generated tables and figures used to regenerate the deterministic
equilibrium checks, reflected-ODE convergence diagnostics, generator-mismatch
calculations, finite-state CTMC drift checks, benchmark comparisons, and the
appendix figures.

\clearpage
\bibliographystyle{plainnat}
\bibliography{bib/references}

\end{document}